%% file: main_new.tex
\documentclass[graybox]{svmult}

\usepackage{mathptmx}       
\usepackage{helvet}         
\usepackage{courier}        
\usepackage{type1cm}        
\usepackage{makeidx}         
\usepackage{graphicx}        
\usepackage{multicol}        
\usepackage[bottom]{footmisc}

\usepackage[T1]{fontenc}
\usepackage[utf8]{inputenc}
\usepackage{latexsym}
\usepackage{times}
\usepackage{microtype}
\usepackage{amssymb}
\usepackage{comment}
\usepackage{mathtools}
\usepackage{enumerate}
\usepackage{array}
\usepackage{multirow}
\usepackage{tabularx}
\usepackage{amsmath}

\usepackage{color}
\usepackage[ruled]{algorithm2e}
\def\qed{\hfill$\Box$\par\vskip1em}

\makeindex             

\begin{document}

\title*{Stand-Up Indulgent Gathering on Lines \\
for Myopic Luminous Robots}
\author{Quentin Bramas
\and Hirotsugu Kakugawa
\and Sayaka Kamei
\and Anissa Lamani
\and \\
Fukuhito Ooshita
\and
Masahiro Shibata
\and 
S\'ebastien Tixeuil
}
\authorrunning{Q. Bramas, H. Kakugawa, S. Kamei, A. Lamani, F. Ooshita, M. Shibata, and S. Tixeuil}

\institute{
Quentin Bramas \at University of Strasbourg, ICube, CNRS, France. \email{bramas@unistra.fr}
\and
Hirotsugu Kakugawa \at Ryukoku University, Japan. \email{kakugawa@rins.ryukoku.ac.jp}
\and
Sayaka Kamei \at Hiroshima University, Japan. \email{s10kamei@hiroshima-u.ac.jp}
\and
Anissa Lamani \at University of Strasbourg, ICube, CNRS, France. \email{alamani@unistra.fr}
\and 
Fukuhito Ooshita \at Fukui University of Technology, Japan. \email{f-oosita@fukui-ut.ac.jp}
\and
Masahiro Shibata\at Kyushu Institute of Technology, Japan. \email{shibata@csn.kyutech.ac.jp}
\and
S\'ebastien Tixeuil\at Sorbonne University, CNRS, LIP6, IUF, France. \email{sebastien.tixeuil@lip6.fr} }

\maketitle

\abstract{We consider a strong variant of the crash fault-tolerant gathering problem called stand-up indulgent gathering (SUIG), by robots endowed with limited visibility sensors and lights on line-shaped networks. In this problem, a group of mobile robots must eventually gather at a single location, not known beforehand, regardless of the occurrence of crashes.
Differently from previous work that considered unlimited visibility, we assume that robots can observe nodes only within a certain fixed distance (that is, they are myopic), and emit a visible color from a fixed set (that is, they are luminous), without multiplicity detection.
We consider algorithms depending on two parameters related to the initial configuration: $M_{init}$, which denotes the number of nodes between two border nodes, and $O_{init}$, which denotes the number of nodes hosting robots. 
Then, a border node is a node hosting one or more robots that cannot see other robots on at least one side. 
Our main contribution is to prove that, if $M_{init}$ or $O_{init}$ is odd, SUIG can be solved in the fully synchronous model. 
}

\section{Introduction}

The Distributed Computing research community actively studies mobile robot swarms, aiming to characterize what conditions make it possible for robots that are confused (each robot has its own ego-centered coordinate system), forgetful (robots may not remember all their past actions) to autonomously move around and solve global problems~\cite{SuzukiY99}. One of these conditions is about how the robots coordinate their actions~\cite{PGN2019}: robots can either act all together (FSYNC), act whenever they want (ASYNC), or act in subsets (SSYNC).

One of the problems that researchers have explored is the \emph{gathering problem}, which serves as a standard for comparison~\cite{SuzukiY99}. It is easy to state (robots must meet at the same place in a finite amount of time, without knowing where it is beforehand), but hard to solve (two robots that move according to SSYNC scheduling cannot meet in finite time~\cite{CourtieuRTU15}, unless there are more assumptions).

Robot failures become more likely as the number of robots increases, or if robots are deployed in dangerous environments, but few studies address this issue~\cite{DefagoPT19}. A crash fault is a simple type of failure, where a robot stops following its protocol unexpectedly. For the gathering problem, the desired outcome in case of crash faults must be specified. There are two options: \emph{weak gathering} requires all non-faulty robots to meet, ignoring the faulty ones, while \emph{strong gathering} (or \emph{stand-up indulgent gathering – SUIG}) requires all non-faulty robots to meet at the unique crash location.  
We believe that SUIG is an attractive task for difficult situations such as dangerous environments: for example, various repair parts could be transported by different robots, and if a robot crashes, the other ones may rescue and repair it after robots carrying relevant parts are gathered at the crash location.

In continuous Euclidean space, weak gathering is solvable in the SSYNC model~\cite{ND2006,ZSS2013,QS2015,XMP2020}, while SUIG (and its variant with two robots, stand up indulgent rendezvous -- SUIR) is only solvable in the FSYNC model~\cite{QAS2020,QAS2021,BramasLT23}. 

Some researchers have recently switched from studying robots in a continuous space to a discrete one~\cite{PGN2019}. In a discrete space, robots can only be in certain locations and move to adjacent ones. This can be modeled by a graph where the nodes are locations, hence the term “robots on graphs”. 
A discrete space is more realistic for modeling physical constraints or discrete sensors~\cite{TPRLSX19}. However, it is not equivalent to a continuous space in terms of computation: a discrete space has fewer possible robot positions, but a continuous space gives more options to resolve difficult situations (e.g., by moving slightly to break a symmetry).

To our knowledge, SUIG in a discrete setting was only considered under the assumption that robots have infinite range visibility (that is, their sensors are able to obtain the position of all other robots in the system that participate to the gathering) in line-shaped networks. Such powerful sensors may seem unrealistic, paving the way for more practical solutions.
With infinite visibility, Bramas et al.~\cite{BramasKLT23} showed that the SUIG problem is solvable in the FSYNC model only.

When infinite range visibility is no longer available, robots become unable to distinguish global configuration situations and act accordingly, in particular, robots may react differently to different local situations, yielding in possible synchronization issues~\cite{KLO14,OPODIS2019}.
In this paper, we consider the discrete setting, and aim to characterize the solvability of the SUIG problem when robots have limited visibility (that is, they are myopic) yet are endowed with visible lights taking colors from a finite set (that is, they are luminous). In particular, we are interested in the trade-off between the visibility range (how many hops away can we see other robots positions) and the memory and communication capacity of the robots (each robot can have a finite number of states that may be communicated to other robots in its visibility range).
In more details, we study SUIG algorithms that depend on two parameters of the initial configuration: $M_{init}$ and $O_{init}$. The former is the number of nodes between two border nodes, and the latter is the number of nodes with robots between two border nodes. A border node has at least one robot and no robots on one or both sides. 
We show that SUIG is solvable in the FSYNC model when either $M_{init}$ or $O_{init}$ is odd.

\section{Model}\label{sec:model}
\input{model_new.tex}

\section{Impossibility Results}

Several impossibility results from the literature hint at which situations are solvable for our problem. Theorem~\ref{T15} and Corollaries~\ref{C4}--\ref{C14} are for the case where robots have no lights.

\begin{theorem}[\cite{REA2008}]\label{T15}
The gathering problem is unsolvable in FSYNC on line
networks starting from an edge-symmetric configuration even if robots can see all the positions of the other robots with global strong multiplicity detection.
\end{theorem}

\begin{corollary}[\cite{BramasKLT23}] \label{C4}
The SUIG problem is unsolvable in FSYNC on line networks
starting from an edge-symmetric configuration even
for robots with infinite visibility and global strong multiplicity detection.
\end{corollary}

\begin{corollary}[\cite{OPODIS2019}] \label{C14}
Starting from a configuration where $M_{init}$ is even and $O_{init}$ is even, there exist initial configurations that a deterministic algorithm cannot gather for myopic robots.
\end{corollary}

As above results, we suppose in the following section that either $M_{init}$ or $O_{init}$ is odd, that is, the initial configurations are not edge-symmetric. The following lemma is also for the case where robots have no lights.

\begin{lemma}[\cite{BramasKLT23}]
Even starting from a configuration that is not edge-symmetric, the SUIG problem is unsolvable in SSYNC 
for robots with infinite visibility and global strong multiplicity detection.
\end{lemma}

\begin{lemma}\label{lem:impossible}
Even starting from a configuration that is not edge-symmetric, the SUIG problem is unsolvable in SSYNC for infinite visibility, global strong multiplicity detection, infinite colors luminous robots.
\end{lemma}
\begin{proof}
Let us suppose for the purpose of contradiction that there exists an algorithm $A$ that solves SUIG for infinite visibility and global strong multiplicity detection luminous robots with an infinite number of colors in SSYNC.
Consider the configuration $C$ that occurs just before gathering is achieved. Now, configuration $C$ has either three consecutive occupied nodes (let us call this configuration class $C_3$) or two consecutive occupied nodes (let us call this configuration class $C_2$). In a configuration in $C_3$, the border robots must be ordered to move inwards by $A$, otherwise gathering is not achieved in the next configuration. From a configuration in $C_3$, the SSYNC scheduler may select only one of the border robots for execution, then reaching a configuration in $C_2$.
So, for any algorithm $A$ that solves $SUIG$, an SSYNC scheduler can reach a configuration in $C_2$. In the sequel, we show that we may never reduce the number of occupied nodes in any execution that starts from a configuration in $C_2$, and hence the gathering is not solved.

Assume that we are in a configuration in $C_2$. Let $k_1$ denote the number of robots on the first occupied node $u_1$, and $k_2$ the number of robots on the second occupied node $u_2$. 
Suppose now that the particular combination of colors at both nodes yields all robots at $u_1$ not to move. Then, we can crash robots at $u_2$. As a result, gathering is never achieved, as the configuration remains in $C_2$ forever. The same argument holds for robots at $u_2$. As a result, algorithm $A$ must command at least one robot at each node to move. 
Now, the scheduler executes those two robots (from the two nodes) that move. Either they both move inwards (exchanging their nodes) and the configuration remains in $C_2$, or at least one of them moves inwards and the resulting configuration remains in $C_2$, or another configuration with more occupied nodes and possibly holes. In any case, the number of occupied nodes is not reduced from two to one, so one can again construct an execution that reaches a configuration in $C_2$, and repeat the argument forever. 
Hence, algorithm $A$ does not solve SUIG, a contradiction.\qed
\end{proof}

As per Lemma~\ref{lem:impossible}, we assume the FSYNC model in the following section.

\section{Possibility Results for Myopic Robots}

\subsection{The case where $M_{init}$ is odd}
In this case, we show that the gathering is achieved even if robots do not have lights. 
For this purpose, in the following we assume that all robots have a single color W (White) which they do not change.
The strategy of our algorithm is as follows:
The robots on two border nodes move towards other occupied node.
The formal description is shown in Algorithm~\ref{alg1}.

\begin{lemma}
Starting from a configuration $C$ where $M_{init}$ is odd, even if there is a crashed robot, all robots gather in $O(M_{init})$ rounds by Algorithm~\ref{alg1}.
\end{lemma}
\begin{proof}
Because we assume FSYNC model, if there is no crashed robot, it is clear that all robots gather on the central node between initial borders in $\lfloor M_{init}/2 \rfloor$ rounds.
If a border robot $r_b$ crashed at time $t<\lfloor M_{init}/2 \rfloor$ on a node $u_i$, it stops at $u_i$. 
If there are other (non-crashed) border robots on $u_i$ at $t$, they move toward other occupied node, thus they are in $u_{i+1}$ at $t+1$.
After that, they cannot move before they become border robots.
On the other hand, the other border robots move towards $u_{i}$.
Thus, eventually, they arrive at $u_{i+1}$ at $M_{init}-2(t+1)+t$-th round, and robots on $u_{i+1}$ become a border.
After that, all robot at $u_{i+1}$ move to $u_i$, and the gathering is achieved.\qed
\end{proof}

\begin{algorithm}[t]
/* Do nothing after gathering. */\\
R0: $\emptyset^\phi[W!]\emptyset^\phi$ :: $\bot$\\

/* Border robots move. */\\
R1: $\emptyset^\phi[W!](\neg(\emptyset^{\phi}))$ :: $\rightarrow$
\caption{Algorithm for the case where $M_{init}$ is odd.}
\label{alg1}
\end{algorithm}

\subsection{The case where $O_{init}$ is odd}

In this section, we show that the gathering is achieved if robots have lights with three colors: W (White), R (Red) and B (Blue).
We assume that all robots have the same color White in the initial configuration.
The formal description is shown in Algorithm~\ref{alg2}.

The transition diagram of configurations by the algorithm in the case that no crash occurs represents in Figs.~\ref{fig:border}--\ref{fig:both3}. In these figures, each small blue box represents a node, and each circle represents the set of robots with the color W, R and B.
The robots represented by doubly lined circles are enabled.
If no crash occurs during the execution, 
the strategy of our algorithm is as follows:
Initially, all robots are White, and robots on two border nodes become Red in the first round by rule R1 (See Fig.~\ref{fig:border}(a)).
The robots on two border nodes move towards other occupied node.
Then, the border robots keep their lights Red or Blue, then the algorithm can recognize that they are border robots.
If there exists a White robot in the adjacent node, the border robot changes its color to Blue or Red by rule R2b or R3b (See Fig.~\ref{fig:border}(b)$\rightarrow$(c),(e)$\rightarrow$(f)).
Otherwise, it just moves without changing its color by rule R2a or R3a (See Fig.~\ref{fig:border}(a),(d)).
When non-border White robots become border, they change their color to Red (resp. Blue) by rule R4a (resp. R4b) if borders that join the node have Red (resp. Blue) (See Fig.~\ref{fig:border}(f)$\rightarrow$(a) or (b) (resp. (c)$\rightarrow$(d) or (e))).
We say that White robots are \emph{captured} by a border if a border moves to the node occupied by the White robots.
When a border node becomes singly-colored, the border robot moves toward other occupied nodes.
Eventually, two borders are neighboring and they have Blue and Red respectively because $O_{init}$ is odd. 
To achieve the gathering, depending on the initial occupied nodes, one of the followings occurs.
\begin{itemize}
\item Case 1: If two borders are singly-colored, the distance between them is two and the central node is empty, both borders move to the central node by rules R2a and R3a (See Fig.~\ref{fig:both}(b))
\item Case 2: If two borders are adjacent and singly-colored, then Blue robots join Red robots by rule R3a at the same time that Red robots become Blue by rule R4c (See Fig.~\ref{fig:both2}(c)).
\item Case 3: If two borders are adjacent, one border has White and Red (resp. Blue) robots and the other border is singly-colored Blue (resp. Red), then White robots become Red (resp. Blue) by rule R4a (resp. R4b) and the singly-colored Blue (resp. Red) border moves to the adjacent border by R3a (resp. R2a) (See Fig.~\ref{fig:one}(b) (resp. (c)).
\item Case 4: If two borders are singly-colored Blue (resp. Red), the distance between them is two and the central node is occupied by White robots, then both borders move to the central node by rule R3b (resp. R2b) (See Fig.~\ref{fig:both3}(b) (resp. (d))).
\end{itemize}
During the execution, if White robot crashes, one border eventually stops executing at the crashed node, but the other border can join the crashed border by rule R2a or R3a.
For the case where Red or Blue robots crash, by special rules R5a--R5c, we are able to respond to various failure patterns.

\begin{algorithm}[t]
{\bf Colors}\\
W (White), R (Red), B (Blue)\\

{\bf Rules}\\
/* Do nothing after gathering. */\\
R0: $\emptyset^\phi[?]\emptyset^\phi$ :: $\bot$\\

/* Start by the initial border robots. */\\
R1: $\emptyset^\phi[W!](\neg\emptyset^{\phi})$ :: $R$\\

/* Border robots on singly-colored nodes move inwards. */\\
R2a: $\emptyset^\phi[R!]
\begin{pmatrix}
\neg W!\\
\neg B!
\end{pmatrix}
(?^{\phi-1})$ :: $\rightarrow$\\

R2b: $\emptyset^\phi[R!](W!)(?^{\phi-1})$ :: $B, \rightarrow$\\

R3a: $\emptyset^\phi[B!](\neg W!)(?^{\phi-1})$ :: $\rightarrow$\\

R3b: $\emptyset^\phi[B!](W!)(?^{\phi-1})$ :: $R, \rightarrow$\\

/* When White robots become border robots, they change their color to the same color as the border robots. */\\
R4a: $\emptyset^\phi\begin{bmatrix}
R\\
[W]
\end{bmatrix}
(?^{\phi})$ :: $R$\\

R4b: $\emptyset^\phi\begin{bmatrix}
B\\
[W]
\end{bmatrix}(?^{\phi})$ :: $B$\\

R4c: $\emptyset^\phi[R!](B!)(\emptyset^{\phi-1})$ :: $B$\\

/* Only for the case that Blue or Red robot crashes. */\\

R5a: $\emptyset^\phi\begin{bmatrix}
R\\
B\\
[W]
\end{bmatrix}(R!,B!)(\emptyset^{\phi-1})$ :: $R$

R5b: $\emptyset^\phi\begin{bmatrix}
B\\
[R]
\end{bmatrix}(R!,B!)(\emptyset^{\phi-1})$ :: $\rightarrow$\\

R5c: $\emptyset^\phi\begin{bmatrix}
R\\
[B]
\end{bmatrix}(R!,B!)(\emptyset^{\phi-1})$ :: $\rightarrow$

\caption{Algorithm for the case where $O_{init}$ is odd. }
\label{alg2}
\end{algorithm}

\begin{figure}[t]
 \begin{minipage}[t]{\linewidth}
 \centering
    \includegraphics[scale=0.5]{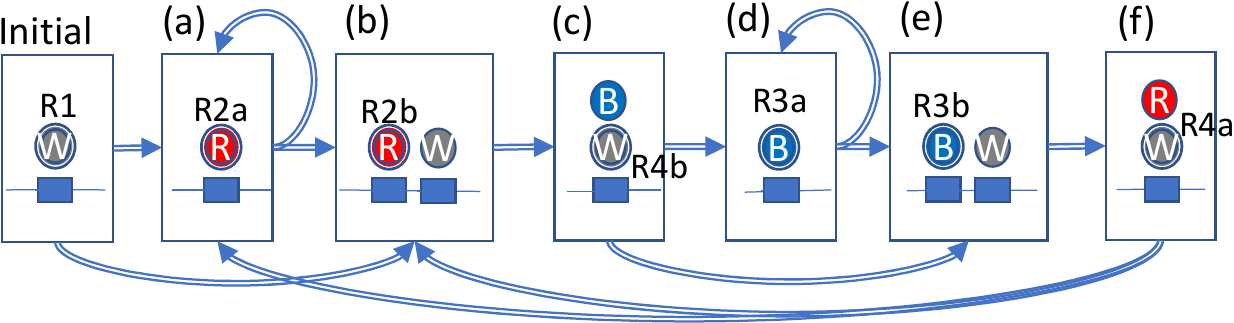}
    \caption{Execution of border robots before the number of occupied nodes becomes three.}
    \label{fig:border}
\end{minipage}\\
\begin{minipage}[t]{0.5\linewidth}
\centering
    \includegraphics[scale=0.5]{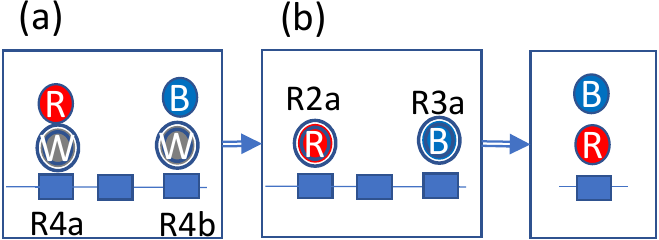}
    \caption{Execution of Case 1.}
    \label{fig:both}
\end{minipage}
\begin{minipage}[t]{0.5\linewidth}
\centering
    \includegraphics[scale=0.5]{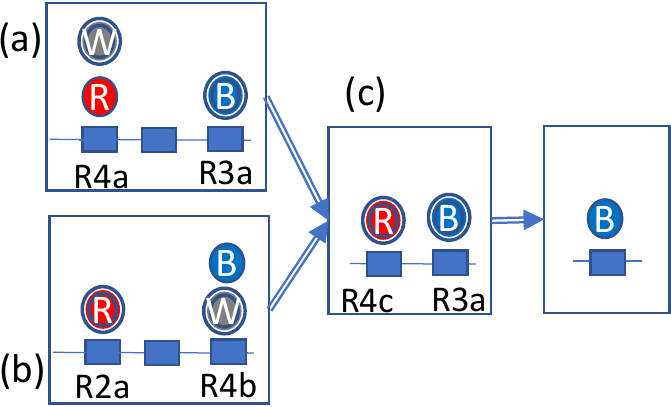}
    \caption{Execution of Case 2.}
    \label{fig:both2}
\end{minipage}\\
\begin{minipage}[t]{\linewidth}
\centering
    \includegraphics[scale=0.5]{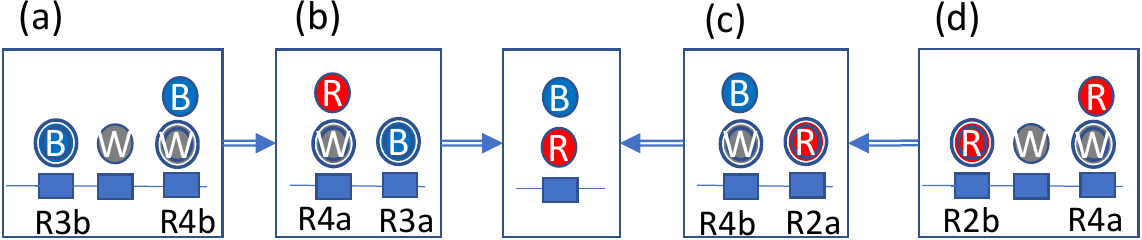}
    \caption{Execution of Case 3.}
    \label{fig:one}
\end{minipage}
\begin{minipage}[t]{\linewidth}
\centering
    \includegraphics[scale=0.5]{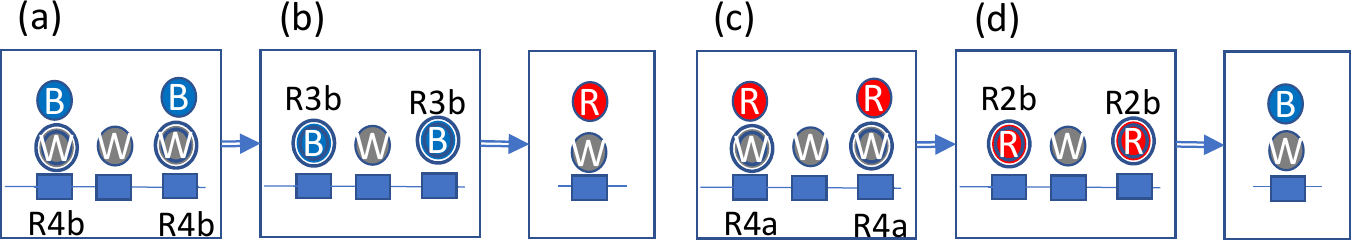}
    \caption{Execution of Case 4.}
    \label{fig:both3}
\end{minipage}
\end{figure}

We prove the correctness of Algorithm~\ref{alg2}. 

\begin{lemma}\label{no-crash}
Starting from a configuration $C$ where $O_{init}$ is odd, if no robot crashes, all robots gather in $O(M_{init})$ rounds by Algorithm~\ref{alg2}.
\end{lemma}
\begin{proof}
Let $r_i$ and $r_j$ be the initial border robots on different sides in the initial configuration $C$. Even if they are towers, we can recognize each of them as a robot because we assume they do not crash in the FSYNC model. 
Because all robots have White color in $C$, only $r_i$ and $r_j$ execute rule R1 and become Red in the first round (Fig.~\ref{fig:border}(a)).
After that, $r_i$ and $r_j$ with Red execute R2a or R2b.
If the adjacent node is occupied by White robots (Fig.~\ref{fig:border}(b)), border robots execute R2b, change their color to Blue and move to the adjacent occupied node (Fig.~\ref{fig:border}(c)). 
Then, the White robots captured by the border execute R4b and the border becomes singly-colored Blue (Fig.~\ref{fig:border}(d) or (e)).
Otherwise, they execute R2a and move to their adjacent node, keeping their color (Fig.~\ref{fig:border}(a)).
After the border robots become Blue, they execute R3a or R3b.
If the adjacent node is occupied by White robot (Fig.~\ref{fig:border}(e)), border robots execute R3b, change their color to Red and move to the adjacent occupied node (Fig.~\ref{fig:border}(f)).
Then, the White robots captured by the border execute R4a and the border becomes singly-colored Red (Fig.~\ref{fig:border}(a) or (b)).
Otherwise, they execute R3a and move to their adjacent node, keeping their color (Fig.~\ref{fig:border}(d)).
Thus, $r_i$ and $r_j$ move toward each other, changing their colors Red and Blue repeatedly when they move to an occupied node.
Note that, borders can only move when they are singly-colored.


Let $t$ be the round when the distance between $r_i$ and $r_j$ becomes two, and $C_t$ be the configuration at $t$.
Let $d_i$ (resp. $d_j$) be the distance that $r_i$ (resp. $r_j$) moved before $t$.
Then, it is clear that $M_{init}-3=d_i+d_j$.
In addition, let $c_i$ (resp. $c_j$) be the number of nodes occupied by White robots such that $r_i$ (resp. $r_j$) captured before $t$.
Then, $c_i+c_j$ is at most $O_{init}$.
Because $O_{init}\leq M_{init}$, $t$ is $O(M_{init})$ rounds.

Consider the execution starting from $C_t$.
First, consider the case that there is no White robot between two borders in $C_t$, i.e., the node between two borders is empty.
Because $O_{init}$ is odd, robots in a border have Red and robots in the other border have Blue. 
\begin{itemize}
\item If both borders are singly-colored in $C_t$, then they move toward each other by R2a and R3a respectively (Fig.~\ref{fig:both}(b)). Then, the gathering is achieved.
\item If both borders include White robots in $C_t$, then White robots in both borders execute R4a or R4b respectively (Fig.~\ref{fig:both}(a)) and both border becomes singly-colored at $t+1$.
\item Consider the case that a border includes White and Red robots and the other has only Blue robots in $C_t$ (Fig.~\ref{fig:both2}(a)). 
At $t+1$, White border robots change their color to Red by R4a and Blue border robots move toward the Red border by R3a.
Then, a singly-colored Red border and a singly-colored Blue border are adjacent (Fig.~\ref{fig:both}(c)).
Then, while Red border robots change their color to Blue by R4c, Blue border robots execute R3a, and the gathering is achieved.
\item Consider the case that a border includes White and  Blue robots and the other has only Red robots in $C_t$ (Fig.~\ref{fig:both2}(b)). 
At $t+1$, White border robots change their color to Blue by R4b and Red border robots move toward the Blue border by R2a.
Then, a singly-colored Red border and a singly-colored Blue border are adjacent (Fig.~\ref{fig:both}(c)).
Then, while Red border robots change their color to Blue by R4c, Blue border robots execute R3a, and the gathering is achieved.
\end{itemize}
Next, consider the case that there is a White robot $r_w$ between two borders in $C_t$.
Because $O_{init}$ is odd, both borders have Blue or both borders have Red.
\begin{itemize}
\item If both borders are singly-colored Red (resp. Blue) at $t$ (Fig.~\ref{fig:both3}(d) (resp. (b))), they moves toward $r_w$ by R2b (resp. R3b). Then, the gathering is achieved.
\item If both borders include White robots at $t$ (Fig.~\ref{fig:both3}(a),(c)), the White border robots change their color to the same color as other border robots by R4a or R4b. Then, we finished the discussion about this case.
\item Consider the case that a border is singly-colored Red robots and the other border includes White robots and Red robots at $t$ (Fig.~\ref{fig:one}(d)).
Then, the White border robots execute R4a and the border becomes singly-colored Red at $t+1$.
At the same time, the singly-colored Red border executes R2b, changes its color to Blue and moves to the node occupied by $r_w$ at $t+1$ (Fig.~\ref{fig:one}(c)).
After that, the singly-colored Red border moves to the other border including $r_w$ by R2a and the gathering is achieved, while $r_w$ changes its color to Blue by R4b. 
\item Consider the case that a border is singly-colored Blue robots and the other border includes White robots and Blue robots at $t$ (Fig.~\ref{fig:one}(a)).
Then, the White border robots execute R4b and the border becomes singly-colored Blue at $t+1$.
At the same time, the singly-colored Blue border executes R3b, changes its color to Red and moves to the node occupied by $r_w$ at $t+1$ (Fig.~\ref{fig:one}(b)). 
After that, the singly-colored Blue border moves to the other border including $r_w$ by R3a and the gathering is achieved, while $r_w$ changes its color to Red by R4a. 
\end{itemize}

Therefore, in any case, the gathering is achieved in $O(M_{init})$ rounds.\qed
\end{proof}

\begin{lemma}\label{lem:White}
Starting from a configuration $C$ where $O_{init}$ is odd, even if a White robot crashes, all robots gather in $O(M_{init})$ rounds by Algorithm~\ref{alg2}.
\end{lemma}
\begin{proof}
We assume that all robots have White initially, and there is no rule such that White robot moves in Algorithm~\ref{alg2}.
Thus, we can discuss the case where the crash of the White robot occurs during the execution by the same way as the case the crash occurs initially.

If a robot $r_k$ at a border $u_i$ of the initial configuration crashes, (1) other robots at $u_i$ becomes Red or (2) the border remains White (i.e., all the robots at $u_i$ crashed).
In both cases, the border at $u_i$ cannot move because $r_k$ remains White.
On the other hand, the robots $r_j$ at the other border change their color to Red by R1 and move toward $u_i$ by R2a or R2b.
By repeating executions of R2a, R2b, R3a and R3b, the border $r_j$ eventually reaches the adjacent node $u_{i+1}$ of $u_i$ (Of course, White robots at other nodes than $u_i$ change their color with $r_j$ by R4a or R4b and move as the border).
\begin{itemize}
\item In the case (1), the border at $u_i$ has Red and White robots, and the other border at $u_{i+1}$ is Blue and White, or singly-colored Blue because $O_{init}$ is odd.
In the former case, the White robots at $u_{i+1}$ change their color to Blue by R4b, and the border at $u_{i+1}$ becomes singly-colored Blue.
Then, the singly-colored Blue border moves to $u_i$ by R3a, and the gathering is achieved.
\item In the case (2), the border at $u_{i+1}$ also becomes singly-colored Blue by the same discussion as above. After that, the Blue border robots at $u_{i+1}$ move to $u_i$ by R3b, and the gathering is achieved.
\end{itemize}

Next, consider the case that a robot at a non-border node of the initial configuration $C$ crashes.
Let $o_1$ be a border node $u_i$, $o_2$ be its neighboring occupied node, $o_k$ be the $k$-th occupied node from $u_i$, $o_{(O_{init})}$ be the other border node $u_{(i+M_{init}-1)}$ in $C$.
Let $r_i$ (resp. $r_j$) be the (sets of) initial border robots at $u_i$ (resp. $u_{(i+M_{init}-1)}$) in $C$.
Without loss of generality, the crash occurs at $o_k$. 
Starting from $C$, both borders move toward $o_k$, and eventually at least one border becomes adjacent to $o_k$.
Let $t$ be the round when at least one border becomes adjacent to $o_k$. 
Without loss of generality, then $r_i$ is adjacent to $o_k$ at $t$.
\begin{itemize}
\item Consider the case that $k$ is odd. 
Then, the border $r_i$ includes Blue robots.
\begin{itemize}
\item Consider the case that $r_j$ is also adjacent to $o_k$ at $t$, and both of $r_i$ and $r_j$ include White robots or both do not include White robots. 
Then, $r_j$ also includes Blue robots because $O_{init}$ is odd.
When both borders include White robots, White border robots execute R4b and change their color to Blue. Thus, both borders become singly-colored Blue.
Then, they move to $o_k$ by R3b, and the gathering is achieved.
\item Consider the case that $r_j$ is also adjacent to $o_k$ at $t$, and one of borders includes White robots.
Without loss of generality, assume that $r_i$ includes White robot at $t$. Then, $r_j$ is singly-colored Blue. 
The White robot occupied with $r_i$ executes R4b at $t+1$.
At the same time, $r_j$ executes R3b, changes its color to Red and moves to $o_k$.
After that, $r_i$ moves to $o_k$ by R3a, and the gathering is achieved. 
\item Consider the case that $r_j$ is not adjacent to $o_k$ at $t$. 
Then, $r_i$ executes R3b, changes its color to Red and moves to $o_k$.
After that, because the White robot on $o_k$ crashes, it cannot change its color.
Thus, the border $r_i$ cannot move from $o_k$.
Eventually, $r_j$ with Blue robots arrives at the adjacent node of $o_k$
(If the node is $o_{k+1}$, the White robots execute R4b and the border $r_j$ becomes singly-colored Blue).
Then, $r_j$ moves to $o_k$ by R3a, and the gathering is achieved.
\end{itemize}
\item Consider the case that $k$ is even.
Then, the border $r_i$ includes Red robots.
\begin{itemize}
\item Consider the case that $r_j$ is also adjacent to $o_k$ at $t$, and both of $r_i$ and $r_j$ include White robots or both do not include White robots. 
Then, $r_j$ also includes Red robots because $O_{init}$ is odd.
When both borders include White robots, White border robots execute R4a and change their color to Red. Thus, both borders become singly-colored Red.
Then, they move to $o_k$ by R2b, and the gathering is achieved.
\item Consider the case that $r_j$ is also adjacent to $o_k$ at $t$, and one of borders includes White robots.
Without loss of generality, assume that $r_i$ includes White robot at $t$.
Then, $r_j$ is singly-colored Red.
The White robot occupied with $r_i$ executes R4a at $t+1$.
At the same time, $r_j$ executes R2b, changes its color to Blue and moves to $o_k$ at $t+1$.
After that, $r_i$ moves to $o_k$ by R2a, and the gathering is achieved. 
\item Consider the case that $r_j$ is not adjacent to $o_k$ at $t$. 
Then, $r_i$ executes R2b, changes its color to Blue and moves to $o_k$.
After that, because the White robot on $o_k$ crashes, it cannot change its color.
Thus, the border $r_i$ cannot move from $o_k$.
Eventually, $r_j$ with Red robots arrives at the adjacent node of $o_k$ (If the node is $o_{k+1}$, the White robots execute R4a and the border $r_j$ becomes singly-colored Red).
Then, $r_j$ moves to $o_k$ by R2a, and the gathering is achieved.
\end{itemize}
\end{itemize}
Thus, the lemma holds.\qed
\end{proof}

\begin{lemma}\label{lem:Red}
Starting from a configuration $C$ where $O_{init}$ is odd, even if a Red robot crashes during the execution, all robots gather in $O(M_{init})$ rounds by Algorithm~\ref{alg2}.
\end{lemma}
\begin{proof}
By the definition of Algorithm~\ref{alg2} and the proof of Lemma~\ref{no-crash}, if a Red robot crashes, it occurs at a border node during the execution. Let $r_k$ be the crashed robot with Red, and $u_k$ be the occupied node by $r_k$. Let $r_j$ be the border robots at the other (non-crashed) border node.
\begin{itemize}
\item Consider the case that the crash occurs just after the execution of R1. 
\begin{itemize}
\item If all robots at $u_k$ crash, the border is singly-colored Red and stops its execution completely.
The other border $r_j$ moves toward $r_k$. 
The neighboring White robot $r_i$ of $r_k$ eventually becomes neighboring to $r_j$, then $r_j$ becomes singly-colored Red. 
After that, $r_j$ moves to the node occupied by $r_i$ by R2b.
Then, $r_i$ changes its color to Blue by R4b.
The Blue border including $r_i$ and $r_j$ moves toward $r_k$ by R3a and eventually moves to $u_k$.
The gathering is achieved.
\item If there is a non-crashed robot $r_i$ in $u_k$ when $r_k$ crashed, then $r_i$ continue to execute, and move toward the other border $r_j$ by R2a or R2b. Let $u_i$ be the node adjacent to $u_k$ where $r_i$ moves. Then, $r_i$ cannot execute any rule before it becomes border.
\begin{itemize}
\item Consider the case that $r_i$ executes R2a. 
Because $O_{init}$ is odd, when $r_j$ reaches to the neighboring White robot of $r_i$, $r_j$ becomes Blue and the border becomes singly-colored Blue.
After that, $r_j$ moves to $u_i$ holding Blue by R3a.
Then, $r_i$ and $r_j$ execute R5b and R5c respectively and move to $u_k$. The gathering is achieved.
\item Consider the case that $r_i$ executes R2b.
Then, $r_i$ is Blue, and there are White robots at $u_i$. 
If $r_j$ moves to $u_i$ at the same time as $r_i$ moves,
White robots at $u_i$ change their color to Blue by R4b, and all robots at $u_i$ move to $u_k$ by R3a. Then, the gathering is achieved.
Otherwise, because $O_{init}$ is odd, when $r_j$ is adjacent to $r_i$, $r_j$ has Red. 
Then, $r_j$ moves to $u_i$ by R2a and $u_i$ has White, Blue and Red robots.
After that, White robot in $u_i$ executes R5a and changes its color to Red, and robots on $u_i$ becomes Blue and Red.
Because $r_k$ is singly-colored Red, robots on $u_i$ execute R5b and R5c, and move to $u_k$.
The gathering is achieved.
\end{itemize}
\end{itemize}

\item Consider the case that the crash occurs just before the execution of R2a. Then, $u_k$ is a border with singly-colored Red robots.
\begin{itemize}
\item Consider the case that its adjacent node $u_{k+1}$ is a border with Blue robots and White robots.  
\begin{itemize}
\item Consider the case that all robots at $u_k$ crash.
Then, White robots at $u_{k+1}$ change their color to Blue by R4b and the border at $u_{k+1}$ becomes singly-colored Blue.
After that, robots at $u_{k+1}$ executes R3a, and the gathering is achieved.
\item Consider the case that there are non-crashed robots $r_i$ at $u_k$.
Then, $r_i$ moves to $u_{k+1}$.
At the same time, White robots at $u_{k+1}$ change their color to Blue by R4b.
Thus, the border $u_{k+1}$ becomes Blue and Red.
In the next round, they moves to $u_k$ by R5b and R5c.
Then, the gathering is achieved.
\end{itemize}
\item If the adjacent node $u_{k+1}$ is empty, we can discuss the same way as the case just after the execution of R1.
\end{itemize}

\item Consider the case that the crash occurs just before the execution of R2b or just after the execution of R2a. Then, we can discuss the same way as the case just after the execution of R1.

\item Consider the case that the crash occurs just before the execution of R4c.
Then, $r_k$ is singly-colored Red border, and the other border $r_j$ is adjacent to $u_k$ and singly-colored Blue.
Thus, $r_j$ moves to $u_k$ by R3a, and the gathering is achieved.

\item Consider the case that the crash occurs just after the execution of R4a.
Then, $r_k$ is a singly-colored Red border, and it is just before the execution of R2a, R2b, or R4c.

\item Consider the case that the crash occurs just after the execution of Compute phase of R3b. 
Then, $r_k$ is adjacent to singly-colored node $u_{k+1}$ with White robots, and $r_k$ changes its color to Red, but does not move. 
\begin{itemize}
\item If all robots at $u_k$ crash, they are Red robots and $u_{k+1}$ is occupied by White robots.
Because $O_{init}$ is odd, when the other borders $r_j$ become adjacent to $u_{k+1}$, then $r_j$ has Blue.
After that, $r_j$ changes its color to Red and moves to $u_{k+1}$ by R3b, and then, White robots at $u_{k+1}$ changes its color to Red by R4a.
Then, $u_{k+1}$ becomes singly-colored Red.
Thus, all robots at $u_{k+1}$ move to $u_k$ by R2a, the gathering is achieved.
\item If there is non-crashed other robots $r_i$ at $u_k$, $r_i$ moves to $u_{k+1}$ with Red color. The other borders $r_j$ move toward $u_{k+1}$.
If $r_j$ moves to $u_{k+1}$ at the same time as $r_i$ moves, White robots at $u_{k+1}$ changes their color to Red by R4a, and all robots at $u_{k+1}$ moves to $u_k$ by R2a. Then, the gathering is achieved.
Otherwise, because $O_{init}$ is odd, when $r_j$ becomes adjacent to $u_{k+1}$, $r_j$ have Blue (Even if there are White robots with them, they eventually become singly-colored Blue by R4b).
After that, $r_j$ moves to $u_{k+1}$ by R3a, then $u_{k+1}$ is occupied by White, Red and Blue robots.
Then, White robots at $u_{k+1}$ changes its color to Red by R5a.
Thus, because all robots at $u_{k+1}$ move to $u_k$ by R5b and R5c, the gathering is achieved.
\end{itemize}

\item Consider the case that the crash occurs just after the execution of Move phase of R3b. Then, $r_k$ moved to an adjacent singly-colored node $u_{k}$ with White robots at round $t$.
\begin{itemize}
\item If $r_j$ is also adjacent to $u_{k}$ and is singly-colored Blue at $t$, then $r_j$ also moves to $u_{k+1}$ by R3b at the same time, and the gathering is achieved.
\item If $r_j$ is adjacent to $u_k$ and is with White robots at $t+1$, White robots in both borders execute R4a and R4b, then both borders becomes singly-colored. Then, $r_j$ moves to $u_k$ by R3a and the gathering is achieved. 
\item If $r_j$ is adjacent to $u_k$ and is singly-colored Blue at $t+1$, then it moves to $u_k$ by R3a. Then, the gathering is achieved.
\item If $r_j$ is not adjacent to $u_k$ at $t+1$, White robots at $u_k$ changes its color to Red by R4a and all robots at $u_k$ becomes Red.
After that, we can discuss this case by the same way as the above cases such that crash occurs in a singly-colored Red border.
\end{itemize}
\end{itemize}
Thus, the lemma holds.\qed
\end{proof}

\begin{lemma}\label{lem:Blue}
Starting from a configuration $C$ where $O_{init}$ is odd, even if a Blue robot crashes during the execution, all robots gather in $O(M_{init})$ rounds by Algorithm~\ref{alg2}.
\end{lemma}
\begin{proof}
By the definition of Algorithm~\ref{alg2} and the proof of Lemma~\ref{no-crash}, if a Blue robot crashes, it occurs at a border node during the execution. Let $r_k$ be the crashed robot with Blue, and $u_k$ be the occupied node by $r_k$. Let $r_j$ be the other (non-crashed) border robots.
\begin{itemize}
\item Consider the case that the crash occurs just before the execution of R3a.
Then, $r_k$ is a singly-colored Blue border.
\begin{itemize}
\item Consider the case that the adjacent node $u_{k+1}$ is a singly-colored Red border.
\begin{itemize}
\item If all robots in $u_k$ crash, Red robots at $u_{k+1}$ change their color to Blue by R4c, and move to $u_k$ by R3a.
Then, the gathering is achieved.
\item If there are non-crashed Blue robots $r_i$ in $u_k$,
$r_i$ executes R3a, and moves to $u_{k+1}$.
At the same time, Red robots at $u_{k+1}$ changes their color to Blue by R4c.
Then, all robots at $u_{k+1}$ are only non-crashed Blue robots, and all robots at $u_k$ are only crashed Blue robots.
Thus, all robots at $u_{k+1}$ moves to $u_k$ by R3a.
Thus, the gathering is achieved.
\end{itemize}
\item Consider the case that the adjacent node $u_{k+1}$ has White and Red robots.
\begin{itemize}
\item If all robots at $u_{k}$ crash, White robots at $u_{k+1}$ change their color by R4a and the border at $u_{k+1}$ becomes singly-colored Red.
After that, the robots at $u_{k+1}$ execute R4c and R3a in sequence, then the gathering is achieved.
\item If there are non-crashed Blue robots $r_i$ in $u_k$,
$r_i$ executes R3a, and moves to the node $u_{k+1}$.
At the same time, White robots at $u_{k+1}$ change their color to Red by R4b.
Thus, the border becomes Blue and Red.
In the next round, they moves to $u_k$ by R5b and R5c, and the gathering is achieved.
\end{itemize}

\item Consider the case that the adjacent node $u_{k+1}$ is empty.
\begin{itemize}
\item If all robots $r_k$ at $u_k$ crash, $r_j$ moves toward $u_k$.
Eventually, $r_j$ becomes adjacent to $u_k$ and singly-colored Red.
After that, $r_j$ becomes singly-colored Blue by R4c, and moves to $u_k$ by R3a.
Then, the gathering is achieved.
\item If there are non-crashed robots $r_i$ at $u_k$ when $r_k$ crashes, $r_i$ moves to the adjacent node $u_{k+1}$ by R3a toward $r_j$.
After that, $r_i$ cannot execute any rule before it becomes border. 
When $r_j$ becomes adjacent to $u_{k+1}$, $r_j$ has Red because $O_{init}$ is odd.
Then, $r_j$ moves to $u_{k+1}$ by R2a, and the border robots at $u_{k+1}$ have Red and Blue robots.
They move to $u_k$ by R5b and R5c, and the gathering is achieved.
\end{itemize}
\end{itemize}

\item Consider the case that the crash occurs just before the execution of R3b. In this case, $r_k$ is also a singly-colored Blue border, and the adjacent occupied node $u_{k+1}$ has singly-colored White robots $r_w$. 
\begin{itemize}
\item Consider the case that all robots at $u_k$ crash, then $r_j$ eventually becomes adjacent to $u_{k+1}$.
Then, $r_j$ has Blue because $O_{init}$ is odd.
If $r_j$ includes White robots, the White robots executes R4b and $r_j$ becomes singly-colored Blue.
Then, $r_j$ executes R3b and moves to $u_{k+1}$.
After that, $r_w$ executes R4a and the border at $u_{k+1}$ becomes singly-colored Red.
Then, robots at $u_{k+1}$ execute R4c and R3a in sequence, and the gathering is achieved. 
\item Consider the case that there are non-crashed border robots $r_i$ in $u_k$ when $r_k$ crashes, $r_i$ continues to execute and moves to $u_{k+1}$ by R3b. Then, $r_i$ become Red and move to $u_{k+1}$ at time $t$. 
Because $O_{init}$ is odd, when $r_j$ is neighboring to $r_w$, $r_j$ has Blue. 
If $r_j$ is also adjacent to $u_{k+1}$ and moves to $u_{k+1}$ at $t$, then we can discuss the case in the same way as above.  Otherwise, $r_i$ cannot execute any rule before it becomes border.
After $r_j$ becomes singly-colored Blue, it moves to $u_{k+1}$ by R3a, and $u_{k+1}$ becomes a border with
White, Red and Blue robots.
After that, White robots $r_w$ at $u_{k+1}$ change their color to Red by R5a, and then, there are only Blue and Red robots at $u_{k+1}$.
Then, all robots at $u_{k+1}$ move to $u_k$ by R5b and R5c, and the gathering is achieved.
\end{itemize}

\item Consider the case that the crash occurs just after the execution of R3a or R4b. Then, $r_k$ is a singly-colored Blue border.
We can discuss it by the same way as the case just before the execution of R3a or R3b.

\item Consider the case that the crash occurs just after the execution of R4c.
By the proof of Lemma~\ref{no-crash}, the gathering is achieved.

\item Consider the case that the crash occurs just after the execution of Compute phase of R2b. Then, just before the execution, the color of $r_k$ is Red, and $r_k$ is adjacent to singly-colored node $u_{k+1}$ with White robots.
After the execution, $r_k$ changes its color to Blue, but does not move.
The other borders $r_j$ move toward $u_{k+1}$.
Because $O_{init}$ is odd, when $r_j$ become neighboring to $u_{k+1}$, $r_j$ have Red (Even if there are White robots with them, they eventually become singly-colored Red by R4a).
\begin{itemize}
\item If all robots at $u_k$ crash,
when $r_j$ moves to $u_{k+1}$ by R2b, then $u_{k+1}$ is occupied by White and Blue robots.
Then, White robots at $u_{k+1}$ change their color to Blue by R4b.
Thus, because all robots at $u_{k+1}$ move to $u_k$ by R3a, the gathering is achieved.
\item If there are non-crashed other robots $r_i$ at $u_k$ just before the execution of R2b, $r_i$ moves to $u_{k+1}$ with Blue color.
If $r_j$ moves to $u_{k+1}$ at the same time as $r_i$ moves, White robots at $u_{k+1}$ change their color to Blue by R4b, and all robots at $u_{k+1}$ move to $u_k$ by R3a. Then, the gathering is achived.
Otherwise, $r_j$ moves to $u_{k+1}$ by R2a, then $u_{k+1}$ is occupied by White, Red and Blue robots.
Then, White robots at $u_{k+1}$ change their color to Red by R5a.
Thus, because all robots at $u_{k+1}$ move to $u_k$ by R5b and R5c, the gathering is achieved.
\end{itemize}

\item Consider the case that the crash occurs just after the execution of Move phase of R2b. Then, $r_k$ moved to an adjacent singly-colored node with White robots. Then, White robots at $u_k$ changes its color to Blue by R4b, and let $t$ be the round.
If $r_j$ is adjacent to $u_k$ at $t$, $r_j$ moves to $u_k$ by R2a and the gathering is achieved at $t$. 
If $r_j$ is not adjacent to $u_k$ at $t$, $u_k$ becomes singly-colored Blue.
After that, we can discuss this case by the same way as the above cases such that crash occurs in a singly-colored Blue border.
\end{itemize}
Thus, the lemma holds.\qed
\end{proof}

From Lemmas \ref{no-crash}--\ref{lem:Blue}, we can deduce: 

\begin{theorem}
Starting from a configuration $C$ where $O_{init}$ is odd, Algorithm~\ref{alg2} solves the SUIG problem on line-shaped networks without multiplicity detection in $O(M_{init})$ rounds.
\end{theorem}

\section{Conclusion}\label{sec:conc}
We presented the first stand-up indulgent gathering algorithms for myopic luminous robots on line graphs. One is for the case where $M_{init}$ is odd, while the other is for the case where $O_{init}$ is odd. 
The hypotheses used for our algorithms closely follow the impossibility results found for the other cases. 

Some interesting questions remain open:
\begin{enumerate}
    \item Are there any algorithms for the case where crashed robots are located at different nodes? (In that case, one has to weaken the gathering specification, e.g., by requiring each correct robot to eventually gather at a crashed location, if any)
    \item Are there any deterministic algorithms for the case where $M_{init}$ and $O_{init}$ are even (Such solutions would have to avoid starting or ending up in edge-view-symmetric situations)?
    \item Are there any algorithms for the case where $O_{init}$ is odd that use fewer colors than ours? 
\item We present distinct solutions for the cases where $M_{init}$ is odd and $O_{init}$ is odd. It would be interesting to design a single algorithm that handles both cases.
\end{enumerate}
\bibliographystyle{spmpsci}
\bibliography{ref}

\end{document}

%% file: model_new.tex
We consider robots that evolve on a line shaped network. The length of the line is infinite in both directions, and consists of an infinite number of nodes $\ldots, u_{-2}, u_{-1}, u_0, u_1, u_2, \ldots$, such that a node $u_i$ is connected to both $u_{(i-1)}$ and $u_{(i+1)}$.

Let $\mathcal{R}=\{r_1, r_2, \dots, r_n\}$ be the set of $n\geq 2$ autonomous robots.
Robots are assumed to be anonymous (i.e., they are indistinguishable), uniform (i.e., they all execute the same program, and use no localized parameter such as a particular orientation), oblivious (i.e., they cannot remember their past actions), and disoriented (i.e., they cannot distinguish left and right).
Then, we assume that robots do \emph{not} know the number of robots.

A node is considered \emph{occupied} if it contains at least one robot; otherwise, it is \emph{empty}. 
If a node contains more than one robot, it is said to have a \textit{tower} or \textit{multiplicity}. 

The \emph{distance} between two \emph{nodes} $u_i$ and $u_j$ is the number of edges between them.
The \emph{distance} between two \emph{robots} $r_p$ and $r_q$ is the distance between two nodes occupied by $r_p$ and $r_q$.
Two robots or two nodes are \emph{adjacent} if the distance between them is one.
Two robots are \emph{neighboring} if there is no robot between them.

Each robot $r_i$ maintains a variable $L_i$, called {\em light}, which spans a finite set of states called \emph{colors}. 
We call such robots \emph{luminous robots}.
A light is \emph{persistent} from one computational cycle to the next: the color is not automatically reset at the end of the cycle. Let $L$ denote the number of available light colors.
Let $L_i(t)$ be the light color of $r_i$ at time $t$. 
We assume the \emph{full light} model: each robot $r_i$ can see the light color of other robots, but also its own light color.
Robots are unable to communicate with each other explicitly (\emph{e.g.}, by sending messages), however, they can observe their environment, including the positions (i.e., occupied nodes) and colors of other robots. 

The ability to detect towers is called \emph{multiplicity detection}, which can be either \emph{global} (any robot can sense a tower on any node) or \emph{local} (a robot can only sense a tower if it is part of it). 
If robots can determine the number of robots in a sensed tower, they are said to have \emph{strong} multiplicity detection.
We assume that robots do \emph{not} have multiplicity detection capability even on their current node but still can sense the visible colors: if there are multiple robots $r_1, r_2,\ldots r_k$ in a node $u$, 
an observing robot $r$ can detect only colors $\{L_i(t)| 1 \leq i \leq k\}$.
So, $r$ can detect there are multiple robots at $u$ if and only if at least two robots among $r_1, r_2,\ldots r_k$ have different colors. 
However, $r$ cannot know how many robots are located in $u$ even if it observe a single color or multiple colors at $u$.

We assume that robots are \emph{myopic}. That is, they have limited visibility: an observing robot $r$ at node $u$ can only sense the robots that occupy nodes within a certain distance, denoted by $\phi$, from $u$. As robots are identical, they share the same $\phi$.

Let $\mathcal{X}_i(t)$ be the set of colors of robots located in node $u_i$ at time $t$. 
If a robot $r_j$ located at $u_i$ takes a snapshot at $t$, the sensor of $r_j$ outputs a sequence, $\mathcal{V}_j$, of $2\phi+1$ sets of colors: 
$$\mathcal{V}_j\equiv\;\mathcal{X}_{i-\phi}(t), \ldots , \mathcal{X}_{i-1}(t), [\mathcal{X}_i(t)], \mathcal{X}_{i+1}(t), \ldots, \mathcal{X}_{i+\phi}(t).$$
This sequence $\mathcal{V}_j$ is the \emph{view} of $r_j$ at $u_i$. 
To distinguish the sequence center, we use square brackets.
If the sequence $\mathcal{X}_{i+1}, \ldots , \mathcal{X}_{i+\phi}$ is equal to the sequence  $\mathcal{X}_{i-1}, \ldots , \mathcal{X}_{i-\phi}$, then the view $\mathcal{V}_j$ of $r_j$ is \emph{symmetric}.
Otherwise, it is \emph{asymmetric}.  
In $\mathcal{V}_j$, a node $u_k$ is \emph{occupied} at time $t$ whenever $|\mathcal{X}_k(t)|>0$.
Conversely, if $u_k$ is \emph{empty} at $t$, 
then $\mathcal{X}_{k}(t)=\emptyset$ holds.

If there exists a node $u_i$ such that $|\mathcal{X}_{i}(t)|=1$ holds, $u_i$ is \emph{singly-colored}.
Note that $|\mathcal{X}_{i}(t)|$ denotes the number of colors at node $u_i$, thus even if $u_i$ is singly-colored, it may be occupied by multiple robots (sharing the same color).
Now, if a node $u_i$ is such that $|\mathcal{X}_{i}(t)|>1$ holds, $u_i$ is \emph{multiply-colored}. As each robot has a single color, a multiply-colored node always hosts more than one robot.

In the case of a robot $r_j$ located at a singly-colored node $u_i$, 
$[\mathcal{X}_i(t)]$ in $r_j$'s view $\mathcal{V}_j$ can be written as $[L_j]$.
Then, without loss of generality, if the left adjacent node of $u_i$ contains one or more robots with color $L_k$, and the right adjacent node of $u_i$ contains one or more robots with color $L_l$, while $u_i$ only hosts $r_j$, then $\mathcal{V}_j$ can be written as $L_k[L_j]L_l$. 
Now, if robot $r_j$ at node $u_i$ occupies a multiply-colored position (with two other robots $r_k$ and $r_l$ having distinct colors), then $|\mathcal{X}_i(t)|=3$, and we can write $\mathcal{X}_i(t)$ in $\mathcal{V}_j$ as 
$\begin{bmatrix}
L_k\\
L_l\\
[L_j]
\end{bmatrix}$.
When the observed node in the view is with multiple colors, we use brackets to distinguish the current position of the observing robot in the view and the inner bracket to explicitly state the observing robot's color.
Note that, because we assume that robots do not have multiplicity detection capability, at $u_i$, there may be two or more robots with $L_k$ and $L_l$ respectively, and there may be two or more robots with $L_j$ other than $r_j$.

Our algorithms are driven by observations made on the current view of a robot, so we use \emph{view predicates}: a Boolean function based on the current view of the robot.
The predicate $L_j$ matches any set of colors that includes color $L_j$, while predicate $(L_j,L_k)$ matches any set of colors that contains $L_j$ or $L_k$. 
Now the predicate  
$\begin{pmatrix}
L_1\\
L_2
\end{pmatrix}$ matches any set that contains \emph{both} $L_1$ and $L_2$. Some of our algorithm rules expect that a node is singly-colored, \emph{e.g.,} with color $L_k$, in that case, the corresponding predicate is denoted by $L_k!$.
To express predicates in a less explicit way, we use character `?' to represent any set, including the empty set. 
The $\neg$ operator is used to negate a particular predicate $P$ (so, $\neg P$ returns \emph{false} whenever $P$ returns \emph{true} and vice versa).
Then, the predicate  
$\begin{pmatrix}
\lnot L_1!\\
\lnot L_2!
\end{pmatrix}$ matches any set that is neither singly-colored $L_1$ nor singly-colored $L_2$.
Also, the superscript notation $P^{y}$ represents a sequence of $y$ consecutive sets of colors, each satisfying predicate $P$. Observe that $y \leq \phi$.
In a given configuration, if the view of a robot $r_j$ at node $u_i$ satisfies predicate ${\emptyset}^{\phi}[?]$ or predicate $[?]{\emptyset}^{\phi}$, then $r_j$ is a \emph{border robot} and $u_i$ a \emph{border node}.

At each time instant $t$, robots occupy nodes, and their positions and colors form a \emph{configuration} $C(t)$ of the system.
Then, each robot $r$ executes Look-Compute-Move cycles infinitely many times: $(i)$ first, $r$ takes a snapshot of the environment and obtains an ego-centered view of the current configuration (Look phase), $(ii)$ according to its view, $r$ decides to move or to stay idle and possibly changes its light color (Compute phase), $(iii)$ if $r$ decided to move, it moves to one of its adjacent nodes depending on the choice made in the Compute phase (Move phase). 
We consider the \emph{FSYNC} model in which at each \emph{round}, each robot $r$ executes an LCM cycle synchronously with all the other robots.
We also consider the \emph{SSYNC} model where a nonempty
subset of robots chosen by an adversarial scheduler executes an LCM cycle synchronously, at each round.
At time instant $t=0$, let $H_{init}$ be the maximum distance between neighboring occupied nodes, $M_{init}$ be the number of nodes between two borders including border nodes, and $O_{init}(\leq M_{init})$ be the number of occupied nodes.
We assume that $\phi \geq H_{init}\geq 1$, i.e., the visibility graph is connected.
As previously stated, no robot is aware of $H_{init}$, $M_{init}$ and $O_{init}$.

In this paper, each rule in the proposed algorithms is presented in the similar notation as in \cite{SF-Ex}: 
$<Label>$ $:$ $<Guard>$ $::$ $<Statement>$. 
The guard is a predicate on the view $\mathcal{V}_j = \mathcal{X}_{i-\phi}, \ldots ,\mathcal{X}_{i-1}, [\mathcal{X}_i],
\mathcal{X}_{i+1}, \ldots , \mathcal{X}_{i+\phi}$
obtained by robot $r_j$ at node $u_i$ during the Look phase. If the predicate evaluates to \emph{true}, $r_j$ is \emph{enabled}, otherwise, $r_j$ is \emph{disabled}.
In the first case, the corresponding rule $<Label>$ is also said to be {\em enabled}.
If a robot $r_j$ is enabled, $r_j$ may change its color and then move based on the corresponding statement during its subsequent Compute and Move phases.
The statement is a pair of ({\it New color}, {\it Movement}).
{\it Movement} can be
($i$) $\rightarrow$, meaning that $r_j$ moves towards node $u_{i+1}$, 
($ii$) $\leftarrow$, meaning that $r_j$ moves towards node $u_{i-1}$, and
($iii$) $\bot$, meaning that $r_j$ does not move.
For simplicity, when $r_j$ does not move (resp. $r_j$ does not change its color), we omit {\it Movement} (resp. {\it New color}) in the statement.
The label $<Label>$ is 
denoted as R followed by a non-negative
integer (\emph{i.e.}, R0, R1, etc.) where a smaller label indicates higher priority.
If the integer in the label is followed by an alphabet (\emph{i.e.}, R1a, R1b, etc.), the priority is determined by the lexicographic order.

\medskip
\noindent\textbf{Problem definition.}
A robot is said to be \emph{crashed} at time instant $t$ if it stops executing at any time $t^\prime \geq t$.
That is, a crashed robot stops execution and remains with the same color at the same position indefinitely.
We assume that robots cannot identify a crashed robot in their snapshots (i.e., they are able to see the crashed robots but remain unaware of their crashed status).
A crash, if any, can occur at any phase of the execution, and break the LCM-atomic (i.e., it can occur the end of round, but also between Look phase and Compute phase or between Compute phase and Move phase).
More than one crash can occur, however we assume that all crashes occur at the same node. 
In our model, since robots do not have multiplicity detection capability, a node with a single crashed robot and with multiple crashed robots with the same color are indistinguishable. Similarly, multiple robots with the same color at the same node have the same behavior, but some or all of them can crash.

We consider the \emph{Stand Up Indulgent Gathering} (SUIG) problem defined in \cite{QAS2021}. 
An algorithm solves the SUIG problem if, for any initial configuration $C_0$ (that may contain multiplicities), 
and for any execution ${\cal E}=(C_0, C_1,\dots)$, there exists a round $t$ such that all robots (including the crashed robot, if any) gather at a single node, not known beforehand, for all $t^\prime\geq t$.
Note that, if there are multiple crashed nodes, the problem cannot be solved. Thus, we need to assume that all the crashes occur at the same node.

Because we assume that robots are anonymous and uniform, all robots have the same color in the initial configuration. 